%% file: journal.tex
\documentclass[a4paper,UKenglish,svgnames]{lipics-v2019}

\input{preamble.tex}
\usepackage{charter}

\excludecomment{shortver}
\includecomment{longver}

\includecomment{cond1}
\newcommand{\longonly}[1]{}
\begin{cond1}
\renewcommand{\longonly}[1]{#1}
\end{cond1}

\excludecomment{cond2}
\newcommand{\shortonly}[1]{}
\begin{cond2}
\renewcommand{\shortonly}[1]{#1}
\end{cond2}

\nolinenumbers
\hideLIPIcs{}

\title{Deleting to Structured Trees}
\titlerunning{Deleting to Structured Trees}

\author{Pratyush Dayal}{Indian Institute of Technology, Gandhinagar, India}{pdayal@iitgn.ac.in}{https://orcid.org/0000-0001-5173-2250}{}
\author{Neeldhara Misra}{Indian Institute of Technology, Gandhinagar, India}{neeldhara.m@iitgn.ac.in}{https://orcid.org/0000-0003-1727-5388}{}
\authorrunning{P. Dayal and N. Misra}

\keywords{Full Binary Trees, Feedback Vertex Set, NP-hardness, Exact Satisfiability}

\begin{document}

\maketitle

\begin{abstract}
  We consider a natural variant of the well-known \textsc{Feedback Vertex Set} problem, namely the problem of deleting a small subset of vertices or edges to a full binary tree. This version of the problem is motivated by real-world scenarios that are best modeled by full binary trees. We establish that both versions of the problem are \NPH{}, which stands in contrast to the fact that deleting edges to obtain a forest or a tree is equivalent to the problem of finding a minimum cost spanning tree, which can be solved in polynomial time. We also establish that both problems are \FPT{} by the standard parameter.
\end{abstract}

\section{Introduction}
\input{intro.tex}

\section{Preliminaries}
\label{sec:prelims}
\input{prelims.tex}

\section{NP-hardness}
\label{sec:nph}
\input{nph.tex}

\section{FPT Algorithms}
\label{sec:fpt}
\input{fpt.tex}


\begin{longver}
\section{Concluding Remarks}
\input{concl.tex}
\end{longver}


\shortonly{\bibliographystyle{plain}}
\bibliography{refs.bib}

\end{document}

%% file: preamble.tex
\usepackage{comment}

\usepackage{graphicx}
\graphicspath{ {./images/} }

\usepackage{float}
\usepackage{wrapfig}
\usepackage{multicol}

\usepackage{tikz}
\usetikzlibrary{positioning}
\usetikzlibrary{arrows}
\usetikzlibrary{decorations.pathmorphing}
\usetikzlibrary{decorations.markings}

\tikzset{
  treenode/.style = {align=center, inner sep=0pt, text centered,
    font=\sffamily},
  arn_t/.style = {treenode, circle, white, font=\sffamily\bfseries, fill=SeaGreen,
   text width=0.7em},
  arn_n/.style = {treenode, circle, white, font=\sffamily\bfseries, draw=black,
    fill=black, text width=0.7em},
  arn_r/.style = {treenode, circle, red, fill=IndianRed,
    text width=0.7em, very thick},
  arn_x/.style = {treenode, rectangle, draw=black,
    minimum width=0.5em, minimum height=0.5em}
        photon/.style={decorate, decoration={snake}, draw=red},
    electron/.style={draw=blue, postaction={decorate},
        decoration={markings,mark=at position .55 with {\arrow[draw=blue]{>}}}},
    gluon/.style={draw=DodgerBlue},
    m1/.style={draw=DarkOrchid, thick},
    m2/.style={draw=Chocolate, thick},
    m3/.style={draw=YellowGreen, thick},
    leaf/.style = {draw=IndianRed, circle},
    root/.style = {draw=SeaGreen, circle},
    internal/.style = {draw=Black, circle},
}

\usepackage[backgroundcolor=Moccasin,colorinlistoftodos]{todonotes}

\makeatletter
 \providecommand\@dotsep{5}
 \def\listtodoname{}
 \def\listoftodos{\@starttoc{tdo}\listtodoname}
 \makeatother

\newcounter{nmcomment}

\newcounter{pdcomment}

\hypersetup{
  colorlinks = true,
  linkbordercolor = {white},
  linkcolor = {IndianRed},
  citecolor = {DarkSeaGreen},
  urlcolor = {MediumVioletRed}
}

\usepackage{latexsym,amssymb,amsmath,mathtools,eulervm,xspace}

\renewcommand{\leq}{\leqslant}

\renewcommand{\bar}[1]{\ensuremath{\overline{#1}}\xspace}

\newcommand{\true}{\textsc{true}\xspace}

\newcommand{\pr}{\ensuremath{\prime}\xspace}

\newcommand{\FPT}{\ensuremath{\mathsf{FPT}}\xspace}

\newcommand{\NPH}{\ensuremath{\mathsf{NP}}-hard\xspace}
\newcommand{\NPC}{\ensuremath{\mathsf{NP}}-complete\xspace}

\newcommand{\CC}{\ensuremath{\mathcal C}\xspace}

\newcommand{\FF}{\ensuremath{\mathcal F}\xspace}

\newcommand{\II}{\ensuremath{\mathcal I}\xspace}

\newcommand{\UU}{\ensuremath{\mathcal U}\xspace}

\newcommand{\MCIS}{\textsc{MCIS}}
\newcommand{\MCISfull}{\textsc{Multi-Colored Independent Set}}

\newcommand{\LES}{\textsc{LNES}}
\newcommand{\LESfull}{\textsc{Linear Near-Exact Satisfiability}}

\newcommand{\FBTDE}{\textsc{FBT-DE}}
\newcommand{\FBTDV}{\textsc{FBT-DV}}

%% file: intro.tex






The \textsc{Feedback Vertex Set} (FVS) problem asks for a smallest subset $S$ of vertices in an undirected graph $G$ to be removed such that the graph, $G \setminus S$, becomes acyclic. This problem was one of the first problems shown to be \NPC{}~\cite{GareyJ79}, and has applications to problems that arise in several areas. These applications include, but are not limited to, operating systems, database systems and VLSI chip design. Consequently, the FVS problem has been widely studied in the context of exact, parameterized and approximation algorithms.

Several variations of the FVS theme have also emerged over the years. For instance, one line of work considers the task of ``deleting to specialized forests'', such as forests of pathwidth one~\cite{CyganPPW12,PhilipRV10} or forests whose connected components are stars of bounded degree~\cite{GanianKO18}. In this case, the forests of pathwidth one turn out to be graphs whose connected components are caterpillars.

Meanwhile, another line of work is the \textsc{Tree Deletion Set} (TDS) problem that considers the issue of the connectivity of the structure after the solution has been deleted. In particular, the TDS problem asks for a smallest subset of vertices $S$ such that $G \setminus S$ is a tree~\cite{GiannopoulouLSS16,RamanSS13}. We remark that the NP-completeness of this TDS problem follows from a general result of Yannakakis~\cite{Yannakakis79}. To state this result, recall that a \emph{property} $\pi$ is a class of graphs, and we will say that $\pi$ is satisfied by, or is true for, a graph $G$ if $G \in \pi$. A property is said to be \emph{non-trivial} if it is satisfied for at least one graph and false for at least one graph; it is \emph{interesting} if the property is true for arbitrarily large graphs and is \emph{hereditary on induced subgraphs} if the deletion of any node from a graph in $\pi$ always results in a graph that is in $\pi$. The result in question states that the problem of finding a maximum connected subgraph satisfying a property $\pi$ is \NPH{} for any non-trivial and interesting property that is hereditary on induced subgraphs.

In this work, we pose a variation of FVS that is in the spirit of a combination of the variations that we have alluded to; here, however, we are looking for a connected object with additional structure. Specifically, we consider the problem of deleting to a full binary tree. We recall that a full binary tree is a tree that has exactly one vertex of degree two and no vertex of degree more than three. Consider the problem of optimally deleting to a full binary tree, posed in the language of the theorem of Yannakakis~\cite{Yannakakis79} stated above, which is to find a maximum connected subgraph that satisfies a certain property. Observe that the property in question could be defined as the property of not having cycles, having exactly one vertex of degree two and no vertex of degree more than three. Note that this property is not hereditary on induced subgraphs: in particular, the deletion of a leaf from a graph that has the property will lead to a violation of the property. In our first result, we explicitly establish the NP-hardness of this problem by reducing from a variant of the \textsc{Independent Set} problem.

In addition, we also consider the edge deletion version of the question above. Recall that for a given connected graph on $n$ vertices and $m$ edges, deleting a smallest subset of edges to obtain a tree is straightforward: it is clear that we have to remove every edge that does not belong to a spanning tree, so the size of the solution is always $(m-(n-1))$. In fact, this problem can be solved in polynomial time even when the edges have weights and we seek a subset of edges of smallest total weight, whose removal results in a tree. It is straightforward to see that any such solution is the complement of a maximum spanning tree and thus, can be found in polynomial time.

In a somewhat surprising twist, we show that the problem of deleting a subset of edges of minimum total weight to obtain a full binary tree is, in fact, \NPC{}. To establish some intuition for why this is true, we briefly sketch a simple reduction from the problem of \textsc{Exact Cover by 3-Sets} to the closely related problem of deleting edges to obtain a full ternary tree.

A ternary tree is a tree where every non-leaf vertex, except the root, is exactly of degree four, while the root has degree three. Let $\FF := \{S_1, \ldots, S_p\}$ be a family of sets of size three over the universe $\UU := \{x_1, \ldots, x_q\}$. The goal is to find a subfamily of disjoint sets whose union is $\UU$. We create a full ternary tree $T$ with $p$ leaves labeled $\{t_1, \ldots, t_p\}$, and set the weight of the edges of $T$ to $B$, a quantity that we will specify later. Then, we introduce for every element $x_i$ in the universe a vertex $v_i$ that is adjacent to $t_j$, if and only if $x_i \in S_j$. The edges between the leaves of $T$ and the vertices corresponding to the elements of $\UU$ have unit weights. We also set $B = 3p-q+1$. It is easy to verify that this graph has a solution of cost $3p-q$ if and only if the system $\UU$ has an exact cover, as desired.

It turns out that establishing the hardness of the problem of deleting to full binary trees is non-trivial, and this is one of our main contributions. We reduce from a fairly restrained version of the \textsc{Satisfiability} problem, the hardness of which is inspired by a reduction in~\cite{ArkinBCCKMS15} and is of independent interest.  We note that we deal with the weighted versions of the problems considered, and we also fix a choice of root vertex as part of the input. Finally, we also note that both the problems we propose above are fixed-parameter tractable, when parameterized by the solution size. To this end, we describe a natural branching algorithm and remark that most preprocessing rules that work in a straightforward manner for \textsc{Feedback Vertex Set} fail when applied as-is to our problem. In particular, it is not trivial to delete degree-one vertices or short-circuit vertices of degree two.


We believe that the problem we propose and the study we undertake has considerable practical motivation. One of the applications of FVS and related problems is to understand noisy datasets. For example, let us say that we expect the data to have a certain structure, but errors in the measurement cause the data at hand not to have the properties expected by said structures. In this context, one approach will be to identify and eliminate the noise - for acyclic structures, that could translate identifying a FVS of small cost. Therefore, for scenarios where the data corresponds to full binary trees, for instances in the case of phylogenetic trees, the problem we present here will be a more relevant model.


%% file: prelims.tex

We follow standard notation and terminology from parameterized complexity~\cite{CyganFKLMPPS15} and graph theory~\cite{diestelbook}; we use $[n]$ to denote the set $\{1,2,\ldots,n\}$. We now turn to the definitions of the problems that we consider. 


\begin{defproblem}{Full Binary Tree Deletion by Vertices (FBT-DV)}{A graph $G=(V,E)$, a vertex $r \in V$, vertex weights $w: V \rightarrow \mathbb{R}^+$, and $k \in \mathbb{Z}^+$.}{Does $G$ have a subset $S \subseteq V$ of total weight at most $k$ such that $G \setminus S$ is a full binary tree rooted at $r$?}
\end{defproblem}

The problems of \textsc{Full Binary Tree Deletion by Edges} (FBT-DE), \textsc{Complete Binary Tree Deletion} (by edges or vertices) and \textsc{Binary Tree Deletion} (by edges or vertices) can be defined analogously. Our focus in this contribution will be mainly on~\FBTDV{} and~\FBTDE{}. 

The \textsc{Multi-Colored Independent Set} problem is the following.

\begin{defparproblem}{Multi-Colored Independent Set (\MCIS)}{A graph $G=(V,E)$ and a partition of $V = (V_1, \ldots, V_k)$ into $k$ parts.}{$k$}{Does there exist a subset $S \subseteq V$ such that $S$ is independent in $G$ and for every $i \in [k]$, $|V_i \cap S| = 1$?}
\end{defparproblem}

%% file: nph.tex

In this section, we establish that the problems of deleting to full binary trees via vertices or edges are \NPC{}. We first describe the hardness for the vertex-deletion variant.

\begin{theorem}
\FBTDV{} is \NPC{}.
\end{theorem}

\begin{proof}
We reduce from \MCISfull{}~\cite[Corollary 13.8]{CyganFKLMPPS15}. Let $(G,k)$ be an instance of \MCIS{} where $G = (V,E)$ and further, let $V = (V_1, \ldots, V_k)$ denote the partition of the vertex set $V$. We assume, without loss of generality, that $|V_i| = n$ for all $i \in [k]$. Specifically, we denote the vertices of $V_i$ by $\{v_1^i, \ldots, v_n^i\}$. We are now ready to describe the reduced instance of FBT-DV, which we denote by $(H,\ell)$.

To begin with, let $H$ be a complete binary tree with $2nk$ leaves, where a complete binary tree is a full binary tree with $2^w$ vertices at distance $w$ from the root for all $w \in [d-1]$, where $d$ the distance between the root and the leaf furthest away from the root. We denote these leaf vertices as: $$\left(\cup_{1 \leq i \leq k}\{a_1^i, \ldots, a_n^i\}\right) \bigcup \left(\cup_{1 \leq i \leq k}\{b_1^i, \ldots, b_n^i\}\right),$$

where, for all $i \in [k]$ and $j \in [n]$, $a_j^i$ and $b_j^i$ are siblings, and their parent is denoted by $p_j^i$.  We refer to this as the \emph{backbone}, to which we will now add more vertices and edges.

For each $i \in [k]$ and $j \in [n]$, we now introduce a third child of $p_j^i$, which we denote by $u_j^i$. We refer to the $u$'s as the \emph{essential} vertices, while its siblings (the $a$'s and the $b$'s) are called \emph{partners}. For all $1 \leq i \leq k$, we also introduce two \emph{guards}, denoted by $x_i$ and $y_i$, which are adjacent to all the essential vertices of type $i$, that is, all $u_j^i$ for $j \in [n]$. Finally, we ensure that the graph induced on the essential vertices is a copy of $G$, more precisely, we have:

$$(u^r_i, u^s_j) \in E(H) \text{ if and only if } (v^r_i, v^s_j) \in E(G) \text{ for all } i \in [k] \text{ and } j \in [n].$$

We set $\ell = nk$. This completes the construction. We now turn to a proof of equivalence.

\paragraph*{The forward direction.} If $S \subseteq V$ is a multi-colored independent set, then consider the subset $S^*$ given by all the essential vertices corresponding to $V \setminus S$, along with the partner vertices $a^i_j$ for each $(i,j)$, for which $v^i_j$ belongs to $S$. It is easy to verify that the proposed set consists of $nk$ vertices. Observe that the deletion of $S^*$ leaves us with a full binary tree where each $p_j^i$ now has two children - either two partner vertices (for vertices not in $S$) or one essential vertex along with one partner vertex (for vertices in $S$). Further, each pair of guards of type $i$ now has an unique parent, which is the essential vertex corresponding to the vertex given by $S \cap V_i$. The essential vertices have degree exactly three because their only other neighbors in $H$ were essential vertices corresponding to neighbors in $G$, but the presence of any such vertex in $H \setminus S^*$ will contradict the fact that $S$ induces an independent set in $G$. This concludes the argument in the forward direction.

\paragraph*{The reverse direction.} Let $S^*$ be a subset of $V(H)$ such that $H \setminus S^*$ is a full binary tree. We claim that $ S^* \cap \{p^i_j \mid 1 \leq i \leq k \text{ and } 1 \leq j \leq n\} = \emptyset$, since the deletion of any parent of a partner vertex will result in the corresponding partner vertex becoming isolated in $H \setminus S^*$--- which leads to a contradiction when we account for the budget constraint on $S^*$. Since all the parents of partner vertices survive and have degree four in $H$, it follows that at least one of its neighbors must belong to $S^*$. In particular, we claim that for every $i \in [k]$ and $j \in [n]$, $ S^* \cap \{u^i_j, a^i_j, b^i_j\} \neq \emptyset$. Indeed, if this is not the case, then $S^*$ contains the parent of $p_i^j$, and it is easy to verify that this leads to a situation where either $H \setminus S^*$ is disconnected or one of the guard vertices has degree two and is not the root, contradicting the assumption that $H \setminus S^*$ is a full binary tree.

From the discussion above, it is clear that $S^*$ picks at least $n$ vertices of type $i$ for each $1 \leq i \leq k$, and combined with the fact that $|S^*| \leq nk$, we note that $S^*$ does not contain any of the guard vertices. Our next claim is that for all $i \in [k]$, $G \setminus S^*$ contains at least one essential vertex of type $i$. If not, then $S^*$ contains all the neighbors of the guards of type $i$, which makes them isolated in $G \setminus S^*$--a contradiction.

For each $1 \leq i \leq k$, consider the vertex in $G$ corresponding to the essential vertex that is \emph{not} chosen by $S^*$ (in the event that there are multiple such vertices, we pick one arbitrarily). We denote this collection of vertices by $S$. We claim that $S$ induces an independent set in $G$: indeed, if not, then any edge in $G[S]$ is also present in $H \setminus S^*$ and creates a cycle when combined with the unique path connecting its endpoints via the backbone, which is again a contradiction. This concludes the proof.
\end{proof}

We now turn our attention to the edge-deletion variant. Here, we will find it convenient to reduce from a structured version of exact satisfiability, where the occurrences of the variables are bounded in frequency and also controlled in terms of how they appear. We will turn to a formal description in a moment, noting that here our reduction is similar to the one used to show that Linear-SAT is \NPC{}~\cite{ArkinBCCKMS15}.

\begin{theorem}
    \label{thm:fbtevnpc}
    \FBTDE{} is \NPC{}.
\end{theorem}

We first describe the version of \textsc{Satisfiability} that we will reduce from. Our instance consists of $(4p+q)$ clauses which we will typically denote as follows:

$$ \CC = \{ A_1, B_1, A_1^\pr, B_1^\pr, \cdots, A_p, B_p, A_p^\pr, B_p^\pr\} \cup \{ C_1, \cdots, C_q \}$$

We refer to the first $4p$ clauses as the \emph{core} clauses, and the remaining clauses as the \emph{auxiliary} clauses. The core clauses have two literals each, and also enjoy the following structure:

$$\forall i \in [p], A_i \cap B_i = \{x_i\} \mbox{ and } A_i^\pr \cap B_i^\pr = \{\bar{x}_i\}$$

We refer to the $x_i$'s as the \emph{main} variables and the remaining variables that appear among the core clauses as \emph{shadow} variables. The shadow variables occur exactly once, and have negative polarity among the core clauses. Therefore, using $\ell(\cdot)$ to denote the set of literals occurring amongst a subset of clauses, we have:

$$ \left| \ell\left( \bigcup_{i = 1}^{p} \{ A_i, B_i, A_i^\pr, B_i^\pr \} \right) \right| = 6p.$$

The auxiliary clauses have the property that they only contain the shadow variables, which occur exactly once amongst them with positive polarity. Also, every auxiliary clause contains exactly four literals. Note that this also implies, by a double-counting argument, that $q = p$. We say that a collection of clauses is a \emph{chain} if it has all the properties described above. An instance of \LESfull{} (\LES{}) is the following: given a set of clauses that constitute a chain, is there an assignment $\tau$ of truth values to the variables such that exactly one literal in every core clause and two literals in every auxiliary clause evaluate to \true{} under $\tau$?

For ease of discussion, given an assignment of truth values $\tau$ we often use the phrase ``$\tau$ satisfies a literal'' to mean that the literal in question evaluates to true under $\tau$. For instance, the question from the previous paragraph seeks an assignment $\tau$ that satisfies exactly one literal in every core clause and two literals in every auxiliary clause. We also refer to such an assignment as a near-exact satisfying assignment. The following observation is a direct consequence of the definitions above.



\begin{proposition}
    \label{prop:brokenchainstructure}
Let $\CC$ be a collection of clauses that form a chain. For any assignment of truth values $\tau$, the main variables satisfy exactly two core clauses and the shadow variables satisfy either one core clause or one auxiliary clause.
\end{proposition}

We first establish that \LES{} is \NPC{}:

\begin{lemma}
\label{lem:lesnph} \LESfull{} is \NPC{}.
\end{lemma}

\begin{proof}
We reduce from (2/2/4)-SAT, which is the variant of \textsc{Satisfiability} where every clause has four literals and every literal occurs exactly twice --- in other words, every variable occurs in exactly two clauses with positive polarity and in exactly two clauses with negative polarity. The question is, if there exists an assignment $\tau$ of truth values to the variables under which exactly two literals in every clause evaluate to true. This problem is known to be \NPC{}~\cite{RatnerW1986}.




Let $\phi$ be a (2/2/4)-SAT instance over the variables $V = \{x_1,\ldots,x_n\}$ and clauses $\CC = \{C_1, \ldots, C_m\}$. For every variable $x_i$, we introduce four new variables: $p_i,r_i$ and $q_i,s_i$. We replace the two positive occurrences of $x_i$ with $p_i$ and $r_i$, and the two negated occurrences of $x_i$ with $q_i$ and $s_i$. We abuse notation and continue to use $\{C_1, \ldots, C_m\}$ to denote the modified clauses. Also, introduce the clauses:\shortonly{~$ A_i = (x_i, \bar{p}_i), B_i = (x_i, \bar{r}_i), A_i^\pr = (\bar{x}_i, \bar{q}_i), B_i^\pr = (\bar{x}_i, \bar{s}_i),$}
\begin{longver}
$$ A_i = (x_i, \bar{p}_i), B_i = (x_i, \bar{r}_i), A_i^\pr = (\bar{x}_i, \bar{q}_i), B_i^\pr = (\bar{x}_i, \bar{s}_i),$$
\end{longver}
for all $1 \leq i \leq n$. Note that these collection of clauses form a chain, as required. We use $\psi$ to refer to this formula. We now turn to the argument for equivalence.

In the forward direction, let $\tau$ be an assignment that sets exactly two literals of every clause in $\phi$ to true. Consider the assignment $\zeta$ given by:

$$ \zeta(x_i) = \tau(x_i), \zeta(p_i) = \zeta(r_i) = \tau(x_i); \zeta(q_i) = \zeta(s_i) = 1 - \tau(x_i),$$

for all $1 \leq i \leq n$. It is straightforward to verify that $\zeta$ satisfies exactly one literal in every core clause and exactly two literals in every auxiliary clause.

In the reverse direction, let $\zeta$ be an assignment for the variables of $\psi$ that satisfies exactly one literal in every core clause and exactly two literals in every auxiliary clause. Define $\tau$ as the restriction of $\zeta$ on the main variables. Let $C$ be a clause in $\phi$. To see that $\tau$ satisfies exactly two literals of $C$, note that the following:

$$\zeta(p_i) = \zeta(r_i) = \zeta(x_i) = \tau(x_i); \zeta(q_i) = \zeta(s_i) = 1 - \zeta(x_i) = 1 - \tau(x_i)$$

is forced by the requirement that $\zeta$ must satisfy exactly one literal in each core clause. Therefore, if $\tau$ satisfies more or less than two literals of any clause $C$, then that behavior will be reflected exactly in the auxiliary clause corresponding to $C$, which would contradict the assumed behavior of $\zeta$. We make this explicit with an example for the sake of exposition. Let $C$ from $\phi$ be the clause $(x_1, \bar{x}_3, \bar{x}_6, x_7)$, and let the clause constructed in $\psi$ be $(p_1, q_3, q_6, r_7)$. Suppose $\tau(x_1) = \tau(x_7) = \tau(x_6) = 1$ and $\tau(x_3) = 0$. Then we have $\zeta(p_1) = \zeta(r_7) = 1$ and $\zeta(q_6) = 0$, while $\zeta(q_3) = 1$. This demonstrates that $\zeta$ satisfies three literals in the auxiliary clause corresponding to $C$, in one-to-one correspondence with the literals that are satisfied by $\tau$. This completes our argument.
\end{proof}


We now turn to a proof of~\autoref{thm:fbtevnpc}. The overall approach is the following. We will introduce a complete binary tree whose leaves will be used to represent variables using variable gadgets which will have obstructions that can be removed in a fixed number of ways, each of which corresponds to a ``signal'' for whether the variable is to be set to true or false. We will then introduce vertices corresponding to clauses that will be attached to the variable gadgets in such a way that they can only be ``absorbed'' into the rest of the structure precisely when exactly two of its literals are receiving a signal indicating that they are being satisfied.

\begin{figure}[t]
    \centering
    \includegraphics[width=0.7\textwidth]{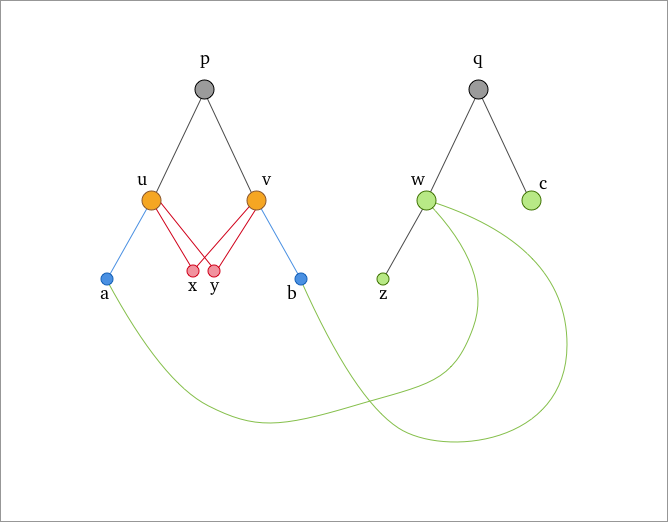}
    \caption{The gadget corresponding to the shadow variables.}
    \label{fig:shadowvargadget}
\end{figure}

\paragraph*{The shadow variables.} An instance of the gadget that we construct for the shadow variables is depicted in~\autoref{fig:shadowvargadget}. We remark here that the notation used for the vertices here is to enable our discussion of how the gadget works and is not to be confused with the notation already used to denote the variables and clauses of the \LES{} instance.

The vertices $p$ and $q$ are called the \emph{anchors} of the gadget, while the vertices $x,y,a$ and $b$ are called the \emph{drivers} of the gadget. This is because, as we will see, the behavior of the edges incident to these vertices determines the fate of the variable --- in terms of whether it ``accepts'' the vertex corresponding to the core clause or the auxiliary clause to which it belongs. We refer to the vertex $u$ in the gadget as the \emph{negative point of entry}, while the vertex $v$ is called the \emph{positive point of entry}.

We refer to the edges incident on the vertices $x,y,a$ and $b$ as \emph{active} edges and the remaining edges (i.e, $(p,u),(p,v),(q,w), (q,c)$ and $(w,z)$) as \emph{passive} edges. We say that a solution is \emph{nice} if it does not contain any passive edges. We also say that an instance $G$ of \FBTDE{} contains a clean copy of the gadget $H$ if $H$ appears in $G$ as an induced subgraph, and further, $d_G(x) = d_G(y) = d_G(a) = d_G(b) = 2$ while $d_G(w) = 4$, $d_G(c) = 1$ and none of the vertices of $H$ are chosen to be the target root vertex. We make the following observation about the behavior of this gadget.

\begin{claim}
    \label{claim:shadowgadget}
    Let $H$ be a vertex gadget for a shadow variable as defined above. Let $G$ be an instance of \FBTDE{} that contains a clean copy of $H$. Then, any nice solution $S$ contains exactly four edges among the edges of $H$.
\end{claim}

\begin{proof}
Let $F$ denote the set of active edges in $G$. Since $d_G(x) = d_G(y) = d_G(a) = d_G(b) = 2$, we claim that any solution $S$ must delete \emph{exactly} four edges from $F$: in particular, $S$ contains exactly one of the edges incident to each of the vertices. Indeed, if $S$ deletes fewer edges than suggested then $G \setminus S$ contains a degree two vertex different from the root, which is a contradiction. On the other hand, if $S$ contains more than four edges from $F$, then at least one of these four vertices is isolated in $G \setminus S$, which contradicts our assumption that $G \setminus S$ is connected. This clearly implies the claim, since all edges not considered are passive and a nice solution does not contain these edges by definition.
\end{proof}

We now analyze the possible behaviors of a solution localized to the gadget in greater detail. \longonly{We refer the reader to~\autoref{fig:shadowgadgetstates} for a visual depiction of the valid states.}\shortonly{We refer the reader to the full version of this paper for the figures associated with this explanation.} \longonly{Based on Claim~\ref{claim:shadowgadget}, we know that in an instance containing a clean copy of the shadow gadget $H$, the projection of any nice solution on $H$ is one of the following:}

\begin{longver}
\begin{multicols}{4}
\begin{enumerate}
\item $xv, yv, au, bw$
\item $xu, yu, aw, bv$
\item $xu, yv, au, bw$
\item $xu, yv, aw, bv$
\item $xv, yu, au, bw$
\item $xv, yu, aw, bv$
\item $xv, yv, aw, bw$
\item $xv, yv, au, bv$
\item $xv, yv, aw, bv$
\item $xu, yu, aw, bw$
\item $xu, yu, au, bv$
\item $xu, yu, au, bw$
\item $xu, yv, au, bv$
\item $xu, yv, aw, bw$
\item $xv, yu, au, bv$
\item $xv, yu, aw, bw$
\end{enumerate}
\end{multicols}
\end{longver}

The possibilities $(xv, yv, aw, bv)$ and $(xu, yu, au, bw)$ do not arise because employing these deletions causes the entry point vertices to have degree four or more in $G \setminus S$. Further, since the solution $S$ does not involve any of the passive edges, then we also rule out the following possibilities, since they all lead to a situation where the degree of $w$ is four or more in $G \setminus S$:

\begin{multicols}{4}
\begin{itemize}
\item $xv, yv, au, bv$
\item $xu, yu, au, bv$
\item $xu, yv, au, bv$
\item $xv, yu, au, bv$
\end{itemize}
\end{multicols}

\begin{longver}
Recalling that $d_G(w) = 4$ when $H$ makes a clean appearance in $G$, we also safely rule out the following possibilities, as they result in a situation where the degree of $w$ is exactly two in $G \setminus S$ --- since $w$ is not the target root vertex, this is a contradiction as well.

\begin{multicols}{4}
\begin{itemize}
\item $xv, yv, aw, bw$
\item $xu, yu, aw, bw$
\item $xu, yv, aw, bw$
\item $xv, yu, aw, bw$
\end{itemize}
\end{multicols}
\end{longver}

\begin{shortver}
Recalling that $d_G(w) = 4$ when $H$ makes a clean appearance in $G$, we also safely rule out the possibilities: $(xv, yv, aw, bw)$, $(xu, yu, aw, bw)$, $(xu, yv, aw, bw)$, $(xv, yu, aw, bw)$. Note that they result in a situation where the degree of $w$ is exactly two in $G \setminus S$ --- since $w$ is not the target root vertex, this is a contradiction as well.
\end{shortver}

Observe that, given a nice solution $S$, in all the valid scenarios possible\longonly{ (also enumerated in~\autoref{fig:shadowgadgetstates})}, either $d_{H \setminus S}(u) = 2$ and $d_{H \setminus S}(v) = 3$, or $d_{H \setminus S}(u) = 3$ and $d_{H \setminus S}(v) = 2$. We say that a shadow variable gadget has a \emph{negative signal} in solutions where $d_{H \setminus S}(u) = 2$. Similarly, we say that the gadget has a \emph{positive signal} in the situations where $d_{H \setminus S}(v) = 2$. We refer to the edges $\{(v,x), (v,y), (u,a), (w,b)\}$ as the \emph{negative witness} and the edges $\{(u,x), (u,y), (v,b), (w,a)\}$ as the \emph{positive witness} for the shadow variable gadgets. This concludes the description of the gadget meant for shadow variables.

\begin{longver}
\begin{figure}[tbhp]
\centering
\includegraphics[width=\textwidth]{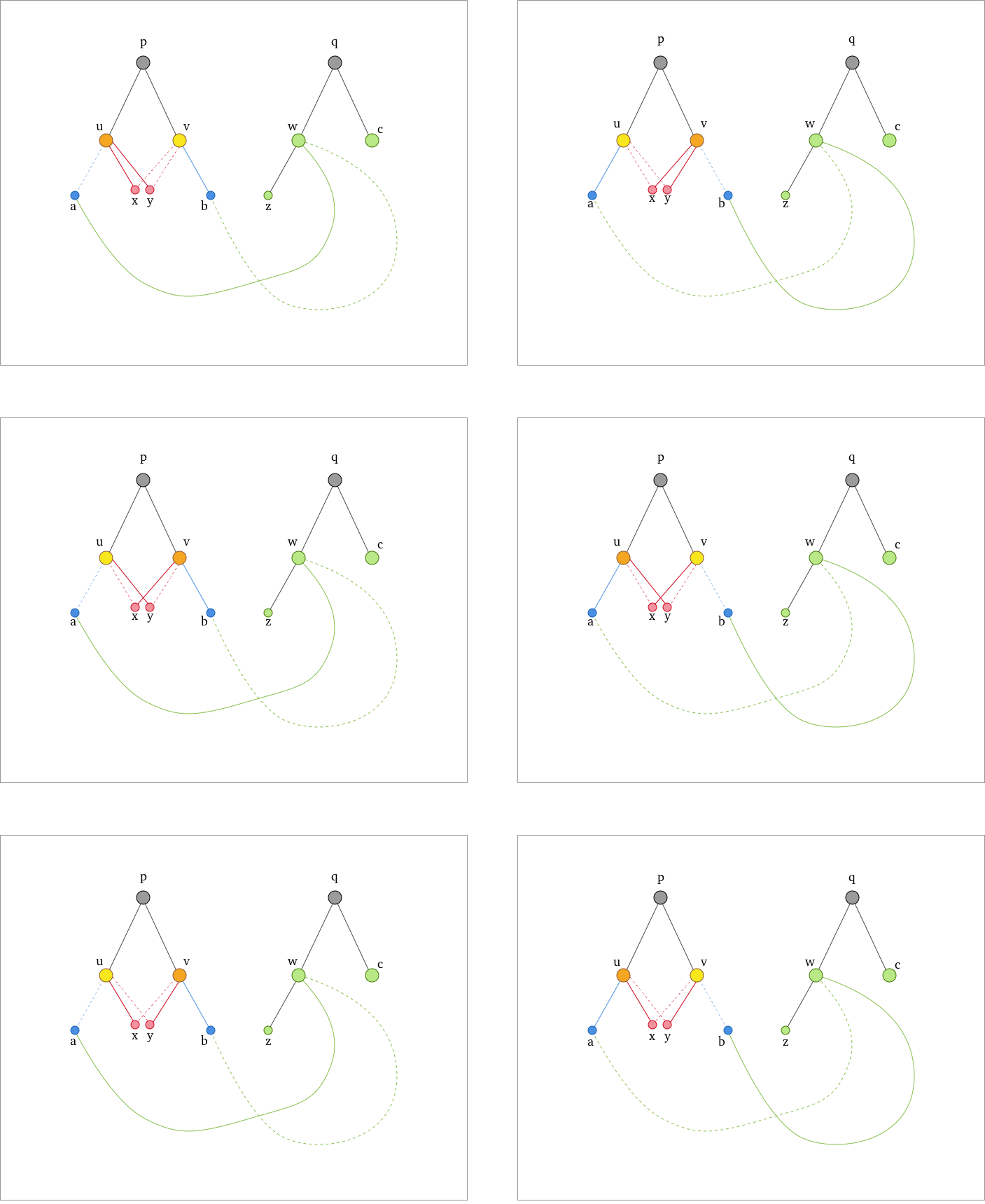}
\caption{The possible states of the gadget corresponding to shadow variables. The deleted edges are depicted by dotted lines. The only degree two vertex remaining after the deletions are made is depicted by a yellow vertex.}
\label{fig:shadowgadgetstates}
\end{figure}
\end{longver}

\paragraph*{The main variables.} We now turn our attention to the gadget corresponding to the main variables. Here, we find it convenient to incorporate vertices representing the core clauses that the main variables belong to also as a part of the gadget. The construction of the gadget is depicted in~\autoref{fig:mainvargadget}. As before, the notation used for the vertices here is to enable our discussion of how the gadget works. With the exception of $A,B, A^\pr, B^\pr$, which indeed are meaningfully associated with the analogously named core clauses, the notation is not to be confused with the notation already used to denote the variables and clauses of the \LES{} instance.

The edges $(z,u)$ and $(z,v)$ are the \emph{passive} edges of this gadget, while the remaining edges are \emph{active}. The vertex $z$ is called the \emph{anchor} of this gadget. As before, a solution is \emph{nice} if it does not contain any of the passive edges. We say that an instance $G$ of \FBTDE{} has a clean copy of $H$ if $H$ appears in $G$ as an induced subgraph and, further, $d_G(p) = d_G(q) = d_G(u) = d_G(v) = 3$, $d_G(x) = d_G(y) = 2$, $d_G(B) = d_G(A^\pr) = 2$, $d_G(A) = d_G(B^\pr) = 3$, and none of the vertices of $H$ are chosen to be the target root vertex.

We reflect briefly on the nature of a nice solution $S$ in instances that have a clean copy of a main variable gadget $H$. First, since $d_G(x) = 2$ and $x$ is not the target root vertex, we note that exactly one of $(v,x)$ or $(u,x)$ must belong to $S$. Suppose $(v,x) \in S$. The removal of $(v,x)$ makes $v$ a vertex of degree two, and since $(z,v)$ is a passive edge and $S$ is nice, $(v,q) \in S$ is forced. Along similar lines, we have that $(u,p) \notin S$. Now, we argue that $(q,B^\pr) \notin S$. Indeed, if $(q,B^\pr) \in S$, then $A^\pr$ has degree two from the deletions so far, and $q$ is a degree-one vertex with $A^\pr$ as its sole neighbor. Recalling that $A^\pr$ is not the target root vertex, we are now forced to delete exactly one of the endpoints incident on $A^\pr$, but both possibilities lead us to a disconnected graph. Therefore, $(q,B^\pr) \notin S$. It is easy to see that this forces $(q,A^\pr) \in S$ and further, $(A,y) \in S$. A symmetric line of reasoning shows that if $(u,x) \in S$, then $(p,u), (p,B)$ and $(B^\pr,y)$ are all in $S$ as well. \longonly{These states are depicted in~\autoref{fig:maingadgetstates}.}\shortonly{We refer the reader to the full version of this paper for the figures associated with this explanation.} These two scenarios motivate the definitions of positive and negative signals that we now make explicit.

With respect to a nice solution $S$, we say that a main variable gadget has a \emph{negative signal} if $d_{H \setminus S}(v) = 2$. Likewise, we say that the gadget has a \emph{positive signal} if $d_{H \setminus S}(u) = 2$. We will also refer to the set of edges $\{(v,x), (v,q), (A^\pr,q), (A,y)\}$ as the \emph{positive witness} of this gadget, and the \emph{negative witness} is defined analogously.

\begin{figure}[t]
    \centering
    \includegraphics[width=0.8\textwidth]{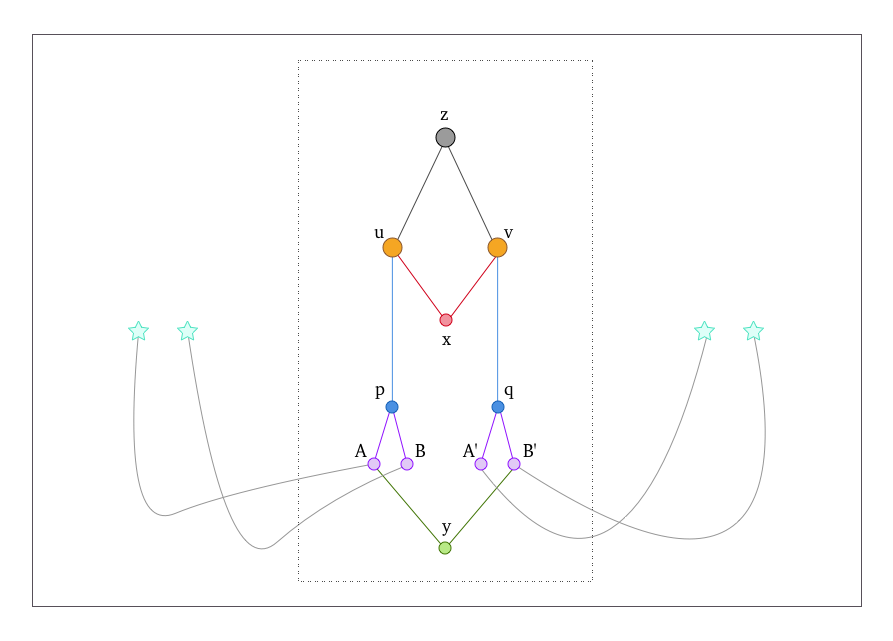}
    \caption{The gadget corresponding to the main variables.}
    \label{fig:mainvargadget}
\end{figure}

\begin{longver}
    \begin{figure}[tbhp]
        \includegraphics[width=.9\textheight,angle=90,origin=c]{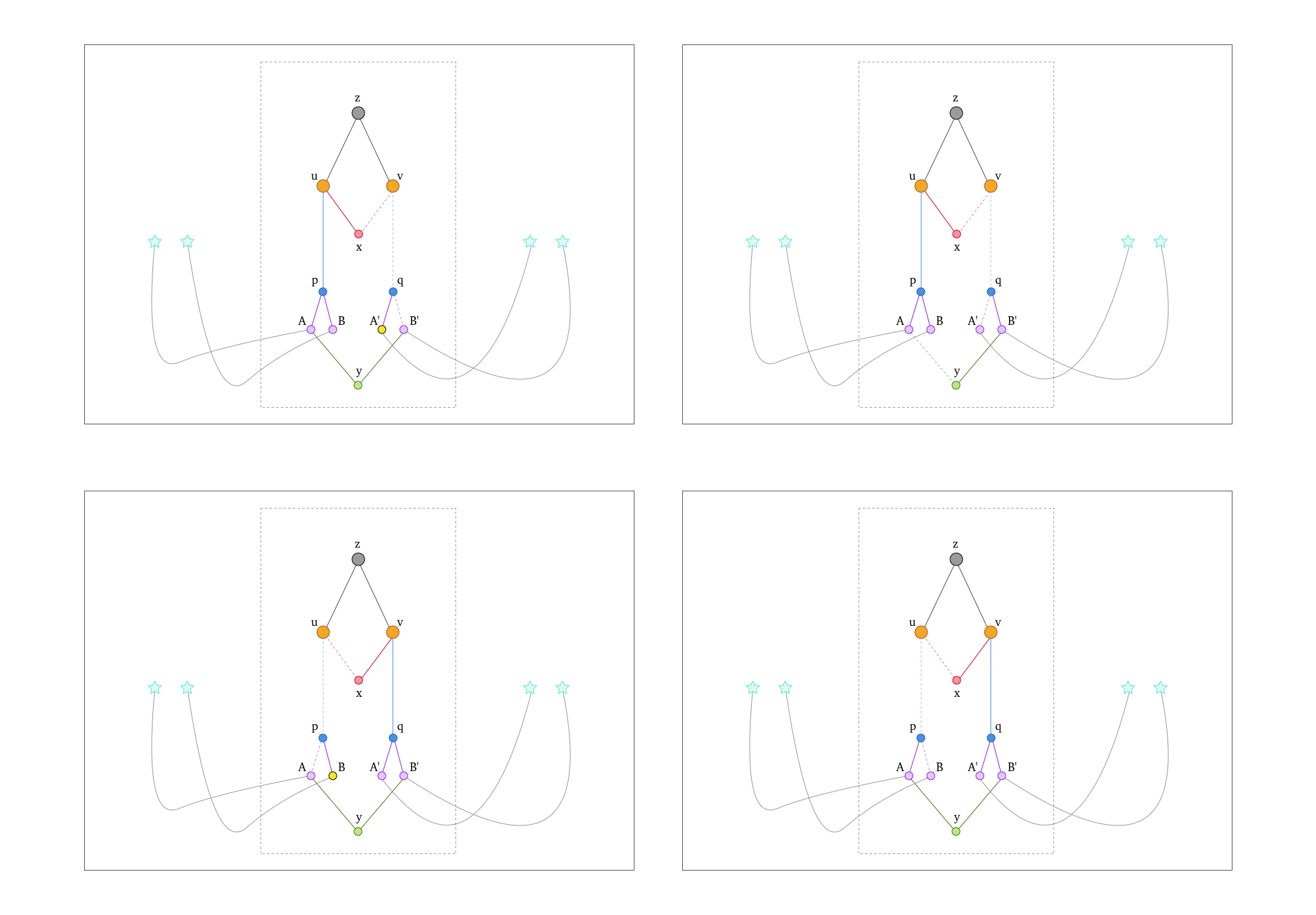}
        \caption{The possible states of the gadget corresponding to main variables. The deleted edges are depicted by dotted lines. The yellow vertices demonstrate the contradictions in the invalid states.}
        \label{fig:maingadgetstates}
    \end{figure}
\end{longver}

We are now ready to discuss the overall construction. Let $\phi$ be an instance of \LES{} with clauses given by:

$$ \CC = \{ A_1, B_1, A_1^\pr, B_1^\pr, \cdots, A_p, B_p, A_p^\pr, B_p^\pr\} \cup \{ C_1, \cdots, C_p \},$$

where the main variable common to $A_i$ and $B_i$ is denoted by $x_i$ and the auxiliary variables in these two clauses are denoted by $p_i$ and $r_i$, while the auxiliary variables in the clauses $A_i^\pr$ and $B_i^\pr$ are denoted $q_i$ and $s_i$. We denote by $\II_\phi := (G,r,w,k)$ the \LES{} instance that we will now construct based on $\phi$.

First, we construct the smallest complete binary tree with at least $(9p)$ leaves and let $\nu$ be the root of this tree. We refer to this tree as the \emph{backbone} of $G$. Let the first $9p$ leaves of this tree be denoted by $\ell_1, \ldots, \ell_p; \alpha_1, \beta_1, \ldots, \alpha_{4p}, \beta_{4p}$. For each main variable $x_i$, let $H_i$ be the corresponding gadget. We identify the anchor of $H_i$ with $\ell_i$. For each shadow variable, we introduce a gadget corresponding to it, and identify the first anchor vertex in the gadget with $\alpha_i$ and the second anchor vertex with $\beta_i$. For every core clause $A_i$, we add an edge between the vertex $A$ in the gadget corresponding to $A_i$ and the negative entry point in the gadget for the shadow variable contained in the clause $A_i$. We also do this in an analogous fashion for the core clauses $A_i^\pr$, $B_i$ and $B_i^\pr$.

Finally, for each auxiliary clause $C_i$, we introduce two vertices $\omega_i$ and $\omega_i^\pr$. We connect these vertices with the positive entry point into all gadgets corresponding to the shadow variables that belong to the clause $C_i$. Note that each of these vertices have degree four. This completes the description of the construction of the graph $G$. We note that all the gadgets present in $G$ are clean by construction. We now define the weight of every edge in the backbone and every passive edge in the gadgets as $(k+1)$, while the weights of the remaining edges are set to be one. Finally, we set $k := 28p$ and let $r = \nu$ --- this concludes the description of the instance $\II_\phi$. \longonly{We now turn to an argument of the equivalence between these instances.}\shortonly{We defer the argument of the equivalence of the instances to the full version of this paper.}

\begin{longver}
\paragraph*{The Forward Direction.} Let $\tau$ be a near-exact satisfying assignment for $\phi$. We now construct an edge deletion set $S$ to a full binary tree as follows:

\begin{enumerate}
    \item If $\tau(x_i) = 1$, then we add the edges corresponding to the positive witness for the gadget for the variable $x_i$ to $S$. Additionally, we include edges incident to to the vertices $A_i$ and $B_i$ in gadget with their other endpoints outside the gadget.
    \item If $\tau(x_i) = 0$, then we add the edges corresponding to the negative witness for the gadget for the variable $x_i$ to $S$. Additionally, we include edges incident to to the vertices $A_i^\pr$ and $B_i^\pr$ in gadget with their other endpoints outside the gadget.
    \item For a shadow variable $y$, if $\tau(y) = 1$, then we include the positive witness from the shadow gadget in $S$. Otherwise, we include the negative witness from the shadow gadget in $S$.
    \item For any auxiliary clause $C_i$, we consider the two literals, say $\ell$ and $\ell^\pr$, that belong to $C_i$ and are satisfied by $\tau$. We delete all edges incident to $\omega_i$ except the one with its other endpoint into the positive entry point of the gadget corresponding to the shadow variable $\ell$. Likewise, we delete all edges incident to $\omega_i^\pr$ except the one with its other endpoint into the positive entry point of the gadget corresponding to the shadow variable $\ell^\pr$.
\end{enumerate}

In the first two steps, we add $6p$ edges, while in the third step we add $4 \cdot (4p) = 16p$ edges to our solution. In the last step, we add three edges for each of the two copies of the $p$ auxiliary clauses, which adds up to a total of $6p$ edges. This gives us a total of $28p$ edges in our solution.

It is straightforward to check that the solution consisting of the prescribed edges indeed forms a valid deletion set that leaves behind a full binary tree: as expected, after the deletions of witness sets from each of the gadgets, only the positive or negative entry points (depending on the nature of the witness) in gadgets corresponding to shadow variables have degree two within the gadget. Consider a pair of shadow variables that appear along with a literal of a given main variable among the core clauses. As a running example, let us consider $p_i$ and $r_i$, appearing along with $x_i$ in the clauses $A_i$ and $B_i$. These variables either satisfy two core clauses or two auxiliary clauses with respect to $\tau$. In the first case, we have $\tau(x_i) = \tau(p_i) = \tau(r_i) = 0$. Observe that when a negative witness is deleted from the gadget corresponding to $x_i$, the vertices corresponding to $A_i$ and $B_i$ in the gadget stay adjacent to the negative entry point in gadgets for $p_i$ and $r_i$, where also a negative witness was removed, making this a mutually compatible set of deletions. In the second situation, $\tau(x_i) = \tau(p_i) = \tau(r_i) = 1$, implying that we removed a positive witness from the gadgets for $p_i$ and $r_i$. Let the two auxiliary clauses satisfied by these shadow variables be $C_j$ and $C_k$ (possibly $j = k$, but assume, without loss of generality, that $j \neq k$). Note that exactly one $\omega_j$ or $\omega_j^\pr$ will be adjacent to the positive entry point in the gadget for the shadow variable $p_i$, and similarly, exactly one $\omega_k$ or $\omega_k^\pr$ will be adjacent to the positive entry point in the gadget for the shadow variable $r_i$. Again, since we removed the positive witness from the gadgets for $p_i$ and $r_i$, this is a mutually compatible set of deletions. A similar argument holds for the shadow variable pair $q_i$ and $s_i$.

\paragraph*{The Reverse Direction.}

Let $S$ be a subset of edges such that $G \setminus S$ is a full binary tree and $|S| \leq 28p$. Observe that $S$ is a nice solution because the passive edges are too expensive relative to the given budget. Recall that every gadget makes a clean appearance in $G$ and therefore, $S$ induces either a positive or a negative signal on each gadget. We define $\tau$ for each variable according to these signals, which is to say that $\tau$ sets a variable to \true{} if and only if $S$ induces a positive signal on the gadget corresponding to the variable.

We also note that every vertex corresponding to auxiliary clauses must have exactly three of their incident edges in $S$. This is because any edge incident on such a vertex has its other endpoint at a positive entry point in a shadow variable gadget. Two of these edges, therefore, when considered along with the path in the backbone joining the relevant entry points, will form a cycle where all other edges are too expensive to belong to $S$. This requires $S$ to contain at least $3 \cdot 2p = 6p$ edges incident on the vertices corresponding to auxiliary clauses. Further, for every main variable gadget, it is easy to verify that $S$ contains two edges that have one endpoint inside the gadget with their other endpoint in a shadow variable gadget. This accounts for another $2p$ edges in the solution. Finally, the positive and negative signals correspond to the deletions of four edges per gadget constructed, and since there are $5p$ gadgets in all ($p$ main variable gadgets and $4p$ shadow variable gadgets), the signals account for the remaining $20p$ units of the budget. This implies that no additional edges other than the ones considered here are deleted by $S$.

We claim that $\tau$ is a near-exact satisfying assignment. First, we consider a core clause $A_i = (\bar{p}_i,x_i)$. If $S$ induced a positive signal on the gadget for $x_i$, then it remains to be seen that $S$ also induces a positive signal on the gadget for $p_i$. Suppose, for the sake of contradiction, that $S$ induces a negative signal on the gadget for $p_i$. Then, the negative entry point in the gadget for $p_i$ has degree two within the gadget (since no other edges are removed by $S$). The only other neighbor of this entry point is the core clause that it satisfies, but the positive signal within the gadget implies that the edges incident to this clause already belongs to $S$ and is therefore --- informally speaking --- not ``avaiable'' for ``fixing'' the degree requirement of the entry point in the shadow variable gadget. Thus, $A_i$ is satisfied exactly if the gadget for $x_i$ has a positive signal. A symmetric argument shows that this is also the case for the negative signal, and the argument works analogously for all the other types of core clauses.

We now turn to the auxiliary clauses. Let $C_i$ be an auxiliary clause, and let $\ell$ and $\ell^\pr$ denote the neighbors of $\omega_i$ and $\omega_i^\pr$ in $G \setminus S$. It is straightforward to infer that for these edges to be compatible with the signals in the gadgets for the shadow variables corresponding to $\ell$ and $\ell^\pr$, it must be the case that the assignment $\tau$ that we have designed, satisfies the literals $\ell$ and $\ell^\pr$ from $C_i$. Further, it is also easy to verify that the other two literals in $C_i$ are \emph{not} satisfied by $\tau$: indeed, if they are, when coupled with the fact that the remaining edges incident on $\omega_i$ and $\omega_i^\pr$ belong to $S$, we will be led to a shadow variable gadget with an entry point that has degree two in $G \setminus S$, which is a contradiction.

This concludes the proof of equivalence of the instances, and in particular, implies the NP-completeness of \FBTDE{}.
\end{longver}


%% file: fpt.tex

We observe that the problems considered here, namely~\FBTDV{} and~\FBTDE{} are fixed-parameter tractable by the standard parameter. We briefly describe a natural branching algorithm for~\FBTDV{} while noting that an analogous argument works for~\FBTDE{}.

Let $(G = (V,E), k, r, w)$ be an instance of~\FBTDV{}. First, consider a vertex $v$, different from the designated root, that has four or more neighbors. Choose any four neighbors of $v$, say $a,b,c,d$, and branch on the set $\{v,a,b,c,d\}$ and we adjust the remaining budget by subtracting the respective weights of these vertices. The exhaustiveness of this branching rule follows immediately from the definition of a full binary tree. Along similar lines, we can also branch on the designated root along with three of its neighbors at a time, if the root has degree three or more, and also the closed neighborhoods of vertices of degree exactly two. We abort any branches where we have exhausted the budget.

We say that a graph with a designated root vertex is \emph{nice} if it is connected, its maximum degree is three and the root is only vertex of degree two. Note that the depth of the branching thus far is bounded by $O(5^k)$, and we branch appropriately on disconnected instances, noting that only one of the components can ``survive''. Also, note that all the remaining instances are nice. If any of the remaining graphs are also acyclic, then we are already done.

If not, then we branch on these graphs further as follows. We pre-process vertices of degree three with a degree one neighbor by employing an appropriate short-circuiting rule. We can then start a breadth-first search (BFS) from the root vertex, noting that the depth of the BFS tree is at most $(\log_2 n + 1)$, since every internal vertex in this tree has at least two children. Therefore, we may infer that there exists a cycle of length $O(\log n)$, and we branch on all the vertices of this cycle other than the root vertex. If the deletion of a vertex on the cycle leads to a disconnected graph, then we abort the corresponding branch. Similarly, if the deletion of a vertex on the cycle creates vertices of degree two in the resulting graph, then we branch on the closed neighborhood of these degree two vertices and discard any disconnected graphs until we arrive at a nice instance, at which point we recurse in the fashion described here. The correctness of the overall algorithm follows from the exhaustiveness of the branching rules. The fact that the running time is FPT follows by a well-known argument~\cite{RamanSS06}.

\begin{theorem}
The problems~\FBTDV{} and~\FBTDE{} are \FPT{} with respect to solution size.
\end{theorem}

%% file: concl.tex

In this contribution, we considered natural variants of the \textsc{Feedback Vertex Set} and the \textsc{Feedback Edge Set} problems with the goal of deleting to structured trees. In particular, we focused our attention on full binary trees with a specified root vertex. We established the NP-hardness of these problems, both in the vertex and edge deletion settings. The latter stands in sharp contrast with the complexity of the analogous problem of deleting edges to obtain a forest (or even a tree), which is equivalent to the problem of finding a Minimum Spanning Forest (respectively, a Minimum Spanning Tree). 

Our work prompts several directions for future work. We believe that it is straightforward to extend our hardness results to the setting of unweighted graphs with no designated root vertex. We also propose considering deletions to other classes of trees, such as complete trees or $p$-ary trees (our focus was on the case $p = 2$). Further, we believe that our \FPT{} approach can be substantially improved by either branching more carefully, or by employing a careful adaptation of the iterative compression approach, which has been successful for both the \textsc{Feedback Vertex Set} and the \textsc{Tree Deletion Set} problems~\cite{RamanSS13}. Finally, we think the kernelization questions for these problems will be very interesting to investigate, possibly building on the ideas that have been used for the kernelization algorithms for \textsc{Feedback Vertex Set} and \textsc{Tree Deletion Set} problems.